\documentclass[11pt]{amsart} \usepackage{amsmath,amssymb,latexsym,amsthm,enumerate}
\usepackage{hyperref} 
\usepackage{graphicx} 
\usepackage{xcolor} 

\theoremstyle{plain}
\newtheorem{theorem}{Theorem}[section]
\newtheorem{proposition}[theorem]{Proposition}
\newtheorem{lemma}[theorem]{Lemma}

\theoremstyle{definition}

\newtheorem{example}[theorem]{Example}
\newtheorem{remark}[theorem]{Remark}

\theoremstyle{remark}

\newcounter{numpar}[section]

\numberwithin{equation}{section}

\newcommand*{\wt}{\widetilde}
\newcommand*{\ol}{\overline}

\newcommand*{\dd}{d}     
\newcommand*{\eps}{\varepsilon}
\newcommand*{\ffi}{\varphi}
\newcommand*{\cA}{\mathcal A}

\newcommand*{\cE}{\mathcal E}
\newcommand*{\cF}{\mathcal F}

\newcommand*{\bbN}{\mathbb N}

\newcommand*{\bbR}{\mathbb R}

\newcommand*{\EE}{\mathsf E}
\newcommand*{\PP}{\mathsf P}
\newcommand*{\QQ}{\mathsf Q}

\newcommand*{\la}{\langle}
\newcommand*{\ra}{\rangle}
\newcommand*{\loc}{{\mathrm{loc}}}

\begin{document}
\title[Delta Hedging in Markets with Jumps]{A Note on Delta Hedging in Markets with Jumps}

\author{Aleksandar Mijatovi\'{c}}
\address{Department of Statistics, University of Warwick, UK}
\email{a.mijatovic@warwick.ac.uk}

\author{Mikhail Urusov}
\address{Institute of Mathematical Finance, Ulm University, Germany}
\email{mikhail.urusov@uni-ulm.de}

\thanks{This research was supported 
by the LMS grant R.STAA.3037 and through the program ``Research in Pairs''
by the Mathematisches Forschungsinstitut Oberwolfach in 2011. 
MU would like to thank the Statistics Department at Warwick, where a part of this work
was carried out, for hospitality and a productive research atmosphere.}



\begin{abstract}
Modeling stock prices via jump processes is common in financial markets.
In practice, to hedge a contingent claim one typically uses the so-called delta-hedging strategy.
This strategy stems from the Black--Merton--Scholes model where it perfectly replicates contingent claims.
From the theoretical viewpoint, there is no reason for this to hold in models with jumps.
However in practice the delta-hedging strategy is widely used 
and its potential shortcoming in models with jumps is disregarded
since such models are typically incomplete and hence most contingent 
claims are non-attainable. 
In this note we investigate a complete model with jumps
where the delta-hedging strategy is well-defined for regular
payoff functions and is uniquely determined via the risk-neutral 
measure. In this setting we give examples of (admissible) 
delta-hedging strategies with bounded discounted value processes, 
which nevertheless fail to replicate the respective bounded contingent claims. 
This demonstrates that the deficiency of the delta-hedging strategy in the presence 
of jumps is not due to the incompleteness of the model but is inherent in 
the discontinuity of the trajectories.
\end{abstract}

\maketitle

\section{Introduction}
\label{sec:i}
The universal success of the Black--Merton--Scholes model (BMS model) is 
in no small measure due to the fact that the risk of a European 
contingent claim can be neutralised using a replication strategy
which, at any moment in time, is a function solely of the current value of the underlying 
security.  This strategy is known as the delta-hedging strategy and is 
universally applied in the financial markets across pricing models.

It is well documented in the literature that modelling 
stock prices as geometric Brownian motions is inconsistent both with the statistical 
properties exhibited by real financial data (``excess kurtosis and skewness'')
and with the option prices observed in the financial markets (``implied volatility smile and skew''),
see e.g. Lai and Xing~\cite{LaiXing:08}, 
Shiryaev~\cite[Ch. IV]{Shiryaev:99}  
for the former, and Gatheral~\cite{Gatheral:06}, and Lipton~\cite{Lipton:01} for the latter.
A widely used class of models, designed to deal with some of these issues, is
processes with jumps (see e.g. Cont and Tankov~\cite{ContTankov:03}, Schoutens~\cite{Schoutens:03}).
The question of how to reduce the outstanding risk of a contingent claim in such 
a model, which typically does not allow perfect replication, 
remains of central importance for market practitioners. 
Even though there are many approaches to hedging in incomplete markets in the literature
(see e.g. the papers Davis, Panas and Zariphopoulou~\cite{DavisPanasZariphopoulou:93}, 
Hubalek, Kallsen and Krawczyk~\cite{HubalekKallsenKrawczyk:06}, 
Schweizer~\cite{Schweizer:01}, Tankov~\cite{Tankov:10},  
the monographs Cont and Tankov~\cite{ContTankov:03}, 
Delbaen and Schachermayer~\cite{DelbaenSchachermayer:06}, 
F\"ollmer and Schied~\cite{FollmerSchied:06}, 
Karatzas and Shreve~\cite{KaratzasShreve:98},
Shiryaev~\cite{Shiryaev:99} and the references therein),
which are applicable in rather general situations, the delta-hedging strategy 
remains the key tool in practice 
because it (1) affords a ``natural'' interpretation based on the BMS model and (2) 
typically possesses a computationally accessible 
implementation 
(one is required to work with functions of time and 
spot price rather than with general predictable processes on some stochastic basis).

Let 
$V(t,x)$ 
be the time-$t$-value (with the stock price at level~$x$) 
of a self-financing strategy that replicates some European contingent claim
in the BMS model.
The \textit{delta-hedging strategy}
is determined uniquely by the requirement
that at each time $t$ a trader following this strategy
holds 
$\frac{\partial V}{\partial x}(t,S_t)$ units of the stock
(and the rest of the wealth in a riskless security), 
where $S_t$ denotes the stock price at time~$t$.
Put differently, the number of units of the stock to be held at time $t$
is computed as a function of $S_t$ by evaluating the partial derivative 
$\frac{\partial V}{\partial x}$
of the value function.

To form an analogous strategy in a model
other than BMS,
where the stock price process $(S_t)_{t\in[0,T]}$ is cadlag,
we first need to understand what the function $V(t,x)$ should be,
since perfect replication is in general no longer possible.
In the financial markets
typically a \emph{pricing measure}~$\PP$ 
is chosen
from a possibly infinite class of equivalent local martingale measures.
The value of a contingent claim is then assumed to be given by the conditional 
expectation of the claim's payoff under the measure~$\PP$.
Let 
$(\cF_t)_{t\in[0,T]}$ 
denote the filtration that models the market information.
If $(S_t)$ is an $(\cF_t,\PP)$-Markov process, 
then the time-$t$-value of a contingent claim $f(S_T)$ is therefore given 
by the function $V(t,x)$, defined by the formula $V(t,S_t)=\EE_\PP(f(S_T)|\cF_t)$.
The key difference with the BMS model is that,
due to the incompleteness of 
$(S_t)$,
in general
$V(t,S_t)$
no longer represents the value of a replicating strategy. 
Nevertheless one can still define the delta-hedging strategy as the self-financing
strategy that starts with the initial capital 
$V(0,S_))$
and holds 
$\frac{\partial V}{\partial x}(t,S_{t-})$ 
units of the stock at each time~$t$
(the argument $S_{t-}$ is necessary to get a predictable strategy).
In general the partial derivative of the value function 
$\frac{\partial V}{\partial x}$ may not exist. 
However, since the delta-hedging strategy 
is widely used in the financial markets,
it is necessary to understand its theoretical properties 
whenever it is well-defined.

The reason why the delta-hedging strategy works in the BMS model is based on the following two arguments:
\smallskip

\noindent (A) integral representation of Brownian martingales implies that any contingent claim is 
replicable;\\
\noindent (B)~\label{p:B} Ito's formula for continuous semimartingales implies that the replicating strategy
should necessarily take the form $\frac{\partial V}{\partial x}(t,S_t)$.
\smallskip

\noindent For general cadlag semimartingales Ito's formula has a different form and even if 
the statement analogous to (A) were true,
there would still be no justification for~(B), i.e. in models with jumps there is no 
reason for the delta-hedging strategy to reduce, let alone eliminate, the remaining risk
of a contingent claim. 
This point concerning~(B) is generally disregarded in practice
because 
in typical models with jumps the analogue of (A) fails to hold.
The reasons why delta hedging 
is used in models with jumps
can be summarized as follows:
\smallskip

\noindent (a) models with jumps applied in financial markets are incomplete
and hence exact replication cannot be required from any hedging strategy;\\
\noindent (b) delta-hedging strategy appears quite naturally because it works in the BMS model,
while the lack of its theoretical justification for models with jumps is attributed to~(a)
alone rather than to the inherent deficiencies of the strategy itself.
\smallskip

\noindent In this paper we consider a simple complete model with jumps 
and demonstrate that the delta-hedging strategy does not typically replicate 
claims that can be replicated by other strategies. This suggests that~(a) should not 
be used to justify the shortcomings of the delta-hedging strategy in general
and that alternative hedging strategies should also be pursued in practice.

The rest of the paper is organised as follows.
Section~\ref{sec:gs}
describes trading strategies and contingent claims in a general semimartingale model for 
the stock price.
In Section~\ref{sec:bsm} we briefly recall the delta-hedging strategy in the BMS model
with the emphasis on related subtle points (such as the nonuniqueness 
of an admissible replicating strategy of a contingent claim).
In Section~\ref{sec:j} we construct the aforementioned complete market with jumps,
where the delta-hedging strategy fails to replicate contingent claims. 
We show that even in such a simple market
model with jumps, an
admissible strategy which replicates a given attainable contingent claim is not unique. 
For a general European payoff, we give an explicit formula for the admissible
replication strategy starting from the minimal possible initial capital.
In the case of a smooth payoff
(which makes the delta-hedging strategy well-defined),
we compare the value process of this  strategy with that of the delta-hedging strategy. 


\section{General Setting}
\label{sec:gs}
Let a discounted stock price be modelled by a semimartingale $S=(S_t)_{t\in[0,T]}$
on a filtered probability space $(\Omega,\cF,(\cF_t)_{t\in[0,T]},\PP)$.
Given a semimartingale $X$, we will use the notation $\Delta X_t:=X_t-X_{t-}$
and the convention $X_{0-}:=X_0$, i.e. $\Delta X_0=0$.

A \emph{strategy} is a pair $\pi=(x_0,\ffi)$,
where $x_0\in\bbR$ and $\ffi=(\ffi_t)_{t\in[0,T]}$ is a predictable $S$-integrable process.
Here $x_0$ is interpreted as the initial capital of this strategy,
and $\ffi_t$ as the number of stocks that a trader following the strategy $\pi$ 
holds at time~$t$.
A \emph{discounted value process} of a strategy $\pi=(x_0,\ffi)$ is
$V^\pi_t=x_0+\int_0^t\ffi_u\,\dd S_u$, $t\in[0,T]$.
A strategy $\pi=(x_0,\ffi)$ is \emph{admissible} if its discounted value process is bounded from below,
i.e. if it holds
\begin{equation}
\label{eq:adm}
\int_0^t\ffi_u\,\dd S_u\ge-c\quad\PP\text{-a.s.},\quad t\in[0,T],
\end{equation}
for some constant $c<\infty$.
Admissibility condition~\eqref{eq:adm} is needed in continuous time finance to preclude arbitrage via
(practically unfeasible) strategies that achieve positive terminal cashflows by allowing unbounded losses in the meantime
(see e.g.~\cite{DelbaenSchachermayer:94}, \cite{DelbaenSchachermayer:06}, \cite{HarrisonPliska:81}, or~\cite{Shiryaev:99}).
In this paper we define a \emph{contingent claim} to be a random variable $f(S_T)$
for some Borel function $f\colon(0,\infty)\to\bbR$ such that
\begin{equation}
\label{eq:claim}
f(S_T)\ge-c\quad\PP\text{-a.s.}
\end{equation}
for some constant $c<\infty$.
Here $f(S_T)$ is interpreted as a discounted payoff at time $T$.
An admissible strategy $\pi=(x_0,\ffi)$ \emph{replicates} a contingent claim $f(S_T)$ if $V^\pi_T=f(S_T)$ $\PP$-a.s.

Random variables that do not satisfy~\eqref{eq:claim} are \emph{not} called contingent claims in this paper
because we are only interested in contingent claims that can be replicated by admissible strategies.
Likewise, we do not consider path-dependent claims here because, in order to speak about delta-hedging strategies,
we only investigate contingent claims that are replicable by strategies with discounted value processes
of the form $V(t,S_t)$ for some function $V(t,x)$, $t\in[0,T]$, $x\in(0,\infty)$.

\section{Delta Hedging in the BMS Model}
\label{sec:bsm}
In this section we briefly recall the notion of delta hedging in the BMS model.
In particular, this will make precise the statements mentioned in the introduction
such as which contingent claims can be replicated,
whether Ito's formula in~(B) on page~\pageref{p:B} can always be applied, etc. 
For general accounts of the BMS
model in the literature on stochastic finance see e.g. the 
monographs~\cite{BinghamKiesel:04}, \cite{Bjork:09}, \cite{DanaJeanblanc:03},
\cite{HuntKennedy:04}, \cite{JeanblancYorChesney:09}, \cite{KaratzasShreve:98},
\cite{LambertonLapeyre:08}, \cite{MusielaRutkowski:08}, \cite{Shiryaev:99},~\cite{Shreve:04}.

Let us consider a standard Brownian motion $W=(W_t)_{t\in[0,T]}$, $0<T<\infty$,
starting from~$0$ on a probability space $(\Omega,\cF,\PP)$,
and let $(\cF_t)_{t\in[0,T]}$ be the right-continuous filtration generated by~$W$.
Let a discounted stock price be given by
$S_t=S_0\exp\{\sigma W_t-(\sigma^2/2)t\}$ ($S_0>0$, $\sigma>0$),
i.e. we have $\dd S_t=\sigma S_t\,\dd W_t$.
In particular, $\PP$ is a risk-neutral measure,
and if $S$ is an $(\cF_t,\QQ)$-local martingale under some $\QQ\sim\PP$,
then $\QQ=\PP$ (i.e. $\PP$ is the unique equivalent risk-neutral measure).

First let us note that if $f(S_T)$ is a contingent claim with $\EE f(S_T)=\infty$,
then it cannot be replicated by an admissible strategy.
Indeed, the discounted value process $V^\pi$ of an admissible strategy $\pi$
is a supermartingale (as a bounded from below local martingale) and 
hence $V^\pi_T\in L^1$, which implies that the $\PP$-a.s.-equality $V^\pi_T=f(S_T)$ cannot hold.
In the rest of the section we therefore consider a contingent claim $f(S_T)$ satisfying~\eqref{eq:claim} and
\begin{equation}
\label{eq:bsm_claim}
\EE f(S_T)<\infty
\end{equation}
(in particular, $f(S_T)\in L^1$).
Let us now discuss how to replicate such a contingent claim~$f(S_T)$. Set
\begin{equation}
\label{eq:U1}
U_t=\EE(f(S_T)|\cF_t),\quad t\in[0,T].
\end{equation}
Since $U$ is an $(\cF_t)$-martingale and the filtration $(\cF_t)$ is generated by the Brownian motion $W$,
there exists a predictable process $\wt\ffi$ with $\int_0^T\wt\ffi_u^2\,\dd u<\infty$ $\PP$-a.s. such that
$U_t=U_0+\int_0^t\wt\ffi_u\,\dd W_u$, $t\in[0,T]$, and hence
\begin{equation}
\label{eq:U2}
U_t=U_0+\int_0^t\ffi_u\,\dd S_u,\quad t\in[0,T],
\end{equation}
where $\ffi_u:=\frac{\wt\ffi_u}{\sigma S_u}$.
In particular, the strategy $\pi:=(U_0,\ffi)$ replicates $f(S_T)$, and $V^\pi_t=U_t$, $t\in[0,T]$.
Finally, by \eqref{eq:claim} and~\eqref{eq:U1}, the process
$U(\equiv V^\pi)$ is bounded from below, hence $\pi$ is admissible.

Let us stress a delicate point related to the replication of contingent claims.
The strategy $\pi=(U_0,\ffi)$ constructed above 
is ubiquitous in the textbooks on stochastic finance
(but at times other argumentation is used to derive it). 
However, it appears to be less appreciated in the literature that
(even in the BMS model!) the same contingent claim $f(S_T)$ can be replicated via different
admissible strategies with different initial capitals. More precisely, we have:
\begin{proposition}
\label{prop:bsm1}
Let $f(S_T)$ satisfy \eqref{eq:claim} and~\eqref{eq:bsm_claim}.

(i) If $\pi=(x_0,\ffi)$ is an admissible strategy replicating $f(S_T)$, then $x_0\ge\EE f(S_T)$.

(ii) For any $x_0\ge\EE f(S_T)$, there is an admissible strategy $\pi=(x_0,\ffi)$ replicating $f(S_T)$.
\end{proposition}

\begin{proof}
(i) Let $\pi=(x_0,\ffi)$ be an admissible strategy replicating $f(S_T)$.
Then $V^\pi$ is a supermartingale as a local martingale bounded from below. Hence, $x_0\ge\EE V^\pi_T=\EE f(S_T)$.

(ii) For the admissible replicating strategy $\pi=(U_0,\ffi)$ with the discounted value process $U$
constructed in \eqref{eq:U1} and~\eqref{eq:U2} we have $U_0=\EE f(S_T)$.
Now let $x_0>U_0$ and set $x=x_0-U_0$. If we find a nonnegative local martingale
$M$ with $M_0=x$ and $M_T=0$ $\PP$-a.s. (which will therefore necessarily be a strict supermartingale), 
then the result follows.
Indeed, by the integral representation theorem, there exists a predictable process $\wt\psi$
such that $\int_0^T \wt\psi_u^2\,\dd u<\infty$ $\PP$-a.s. and
$M_t=x+\int_0^t\wt\psi_u\,\dd W_u=x+\int_0^t\psi_u\,\dd S_u$, $t\in[0,T]$,
where $\psi_u:=\frac{\wt\psi_u}{\sigma S_u}$.
Then the strategy $(x_0,\ffi+\psi)$ is admissible and replicates~$f(S_T)$.

In order to find such a nonnegative local martingale $M$, we first consider the continuous local martingale
$L=(L_t)_{t\in[0,T)}$ on $[0,T)$ defined by the formula
$$
L_t=x+\int_0^t\frac1{T-u}\,\dd W_u,\quad t\in[0,T).
$$
By the Dambis--Dubins--Schwarz theorem, $L_t=B_{\la L,L\ra_t}$, $t\in[0,T)$,
for some Brownian motion $B$ starting from~$x$.
Since $\la L,L\ra_t=\frac1{T-t}-\frac1T\nearrow\infty$ as $t\nearrow T$,
we get
$$
\limsup_{t\nearrow T}L_t=-\liminf_{t\nearrow T}L_t=\infty\quad\PP\text{-a.s.},
$$
hence
$$
\tau:=\inf\{t\in[0,T)\colon L_t=0\}<T\quad\PP\text{-a.s.}
$$
It remains to set $M=L^\tau$.
\end{proof}

Thus, any contingent claim $f(S_T)$ satisfying \eqref{eq:claim} and~\eqref{eq:bsm_claim}
can be replicated via different admissible strategies with different
initial capitals, but the discounted value process $U$ of the strategy $\pi=(U_0,\ffi)$
constructed in \eqref{eq:U1} and~\eqref{eq:U2} starts from the minimal possible initial capital.
Now we are going to prove that this strategy $\pi$ has to coincide with the delta-hedging 
strategy.

We have
\begin{equation}
\label{eq:bsm1}
U_t=\EE(f(S_T)|\cF_t)=\EE\left(f\left(S_t e^{\sigma(W_T-W_t)-\frac{\sigma^2}2(T-t)}\right)|\cF_t\right)=V(t,S_t)
\quad\PP\text{-a.s.},
\end{equation}
where
$$
V(t,x):=\EE f\left(x e^\xi\right),\quad t\in[0,T],\;x\in(0,\infty),
$$
with $\xi\sim N(-\sigma^2(T-t)/2,\sigma^2(T-t))$.
The last equality in~\eqref{eq:bsm1} follows from the facts that
$S_t$ is $\cF_t$-measurable and $W_T-W_t$ is independent of~$\cF_t$.
Let us set
\begin{align*}
\ol f(y)&=f\left(e^y\right),\quad y\in\bbR,\\
\ol V(t,y)&=V\left(t,e^y\right),\quad t\in[0,T],\;y\in\bbR.
\end{align*}
Then
$$
\ol V(t,y)=\EE\ol f(y+\xi)=\frac1{\sqrt{2\pi(T-t)}\sigma}\int_\bbR
\ol f(z)e^{-\frac{(z-y+\sigma^2(T-t)/2)^2}{2\sigma^2(T-t)}}\,\dd z.
$$
It follows that $\ol V\in C^{1,2}([0,T)\times\bbR)$ and hence $V\in C^{1,2}([0,T)\times(0,\infty))$.

By Ito's formula,
\begin{equation}
\label{eq:bsm2}
U_t=V(t,S_t)=U_0+\int_0^t\frac{\partial V}{\partial x}(u,S_u)\,\dd S_u+A_t,\quad t\in[0,T),
\end{equation}
where $A$ is a continuous process of finite variation.
The stochastic integrals with respect to $S$ in \eqref{eq:U2} and~\eqref{eq:bsm2}
are continuous local martingales. Thus, $A\equiv0$ and the replicating strategy $\pi=(U_0,\ffi)$
is the delta-hedging strategy: $\ffi_t=\frac{\partial V}{\partial x}(t,S_t)$, $t\in[0,T)$
(note that the terminal value $\ffi_T$ is, in fact, irrelevant).

\begin{remark}
\label{rem:bsm1}
The argument above also explains that it is essentially Ito's formula
for continuous semimartingales that justifies delta hedging in the BMS model,
hence there is no reason for delta hedging to be a sensible strategy in models with jumps.
Indeed, let us now suppose that $S$ is a cadlag martingale
(or local martingale, or sigma-martingale; see~\cite{DelbaenSchachermayer:98}),
so that we again start with a risk-neutral measure from the outset.
Further let us consider a bounded from below integrable contingent claim $f(S_T)$ and assume that
\begin{equation}
\label{eq:bsm3}
U_t:=\EE(f(S_T)|\cF_t)=U_0+\int_0^t\ffi_u\,\dd S_u,\quad t\in[0,T],
\end{equation}
for some strategy~$\ffi$ (which is then automatically admissible), and that
$$
U_t=V(t,S_t),\quad t\in[0,T],
$$
for some $C^{1,2}$-function $V(t,x)$, $t\in[0,T]$, $x\in\bbR$.
By Ito's formula we find
\begin{align}
\label{eq:bsm4}
V(t,S_t)&=U_0+\int_0^t\frac{\partial V}{\partial x}(u,S_{u-})\,\dd S_u
+\int_0^t\frac{\partial V}{\partial t}(u,S_{u-})\,\dd u
+\frac12\int_0^t\frac{\partial^2 V}{\partial x^2}(u,S_{u-})\,\dd\la S^c,S^c\ra_u\\
\notag
&+\sum_{u\le t}\left(V(u,S_u)-V(u,S_{u-})-\frac{\partial V}{\partial x}(u,S_{u-})\Delta S_u\right),
\end{align}
where $S^c$ denotes the continuous martingale part of~$S$ (see e.g. \cite[Ch.~I, Prop.~4.27 and Th.~4.57]{JacodShiryaev:03}).
In contrast to the BMS model, comparison of \eqref{eq:bsm3} and~\eqref{eq:bsm4}
does not imply that $\ffi_t$ equals $\frac{\partial V}{\partial x}(t,S_{t-})$
because there may be a nontrivial martingale part in the jump term in the right-hand side of~\eqref{eq:bsm4}.
\end{remark}

\section{Hedging in a Complete Market with Jumps}
\label{sec:j}
Let us consider a standard Poisson process $N=(N_t)_{t\in[0,T]}$, $0<T<\infty$, starting from~$0$
with an intensity $\lambda>0$ on a probability space $(\Omega,\cF,\PP)$,
and let $(\cF_t)_{t\in[0,T]}$ be the right-continuous filtration generated by~$N$,
i.e. $\cF_t=\bigcap_{\eps>0}\sigma(N_u;\,u\in[0,t+\eps])$.
We set $X_t=N_t-\lambda t$, so $X$ is a martingale.
Let us fix $S_0>0$, $\sigma>0$ and consider the discounted stock price $(S_t)_{t\in[0,T]}$
driven by the SDE
$$
\dd S_t=\sigma S_{t-}\,\dd X_t,
$$
i.e.
$$
S_t=S_0\cE(\sigma X)_t=S_0e^{\alpha N_t-\beta t}
$$
with
\begin{equation}
\label{eq:j_a1}
\alpha=\log(1+\sigma)>0,\quad\beta=\sigma\lambda>0.
\end{equation}
In particular, $S$ is a strictly positive  $(\cF_t,\PP)$-martingale.
It follows from \cite[Th.~II.4.5]{JacodShiryaev:03} and the fact that
$(\cF_t)$ is generated by $N$ that if $X$ (or equivalently,~$S$)
is an $(\cF_t,\QQ)$-local martingale under some $\QQ\sim\PP$, then $\QQ=\PP$
(i.e. $\PP$ is the unique equivalent risk-neutral measure).

Like in the previous section, a contingent claim $f(S_T)$ with $\EE f(S_T)=\infty$
cannot be replicated by an admissible strategy. Therefore, below we always consider
a contingent claim $f(S_T)$ satisfying~\eqref{eq:claim} and
\begin{equation}
\label{eq:j_claim}
\EE f(S_T)<\infty
\end{equation}
(in particular, $f(S_T)\in L^1$).
In order to replicate such a contingent claim let us define the $(\cF_t)$-martingale
\begin{equation}
\label{eq:j_d1}
U_t=\EE(f(S_T)|\cF_t),\quad t\in[0,T].
\end{equation}
In particular, $U_0=\EE f(S_T)$.
By \cite[Th.~III.4.37]{JacodShiryaev:03}, there is a predictable process $\wt\ffi$
with\footnote{As in~\cite{JacodShiryaev:03}, we deonote by $\cA^+_\loc$ 
the class of nondecreasing adapted cadlag processes $A$ with $A_0=0$
such that there exists a nondecreasing sequence of stopping times $(\tau_n)$
with $\{\tau_n=T\}\nearrow\Omega$~$\PP$-a.s. and $\EE A_{\tau_n}<\infty$.}
$\int_0^\cdot|\wt\ffi_u|\,\dd N_u\in\cA^+_\loc$
(or equivalently, $\int_0^T|\wt\ffi_u|\,\dd u<\infty$ $\PP$-a.s.) such that
\begin{equation}
\label{eq:j_d2}
U_t=U_0+\int_0^t\wt\ffi_u\,\dd X_u,\quad t\in[0,T],
\end{equation}
hence
\begin{equation}
\label{eq:j_d3}
U_t=U_0+\int_0^t\ffi_u\,\dd S_u,\quad t\in[0,T],
\end{equation}
where $\ffi_u:=\frac{\wt\ffi_u}{\sigma S_{u-}}$.
Thus, $\pi:=(U_0,\ffi)$ is an admissible strategy replicating $f(S_T)$.
Furthermore, we have $U_t=V^\pi_t$, $t\in[0,T]$.

One may wonder whether in this simple model based on the Poisson process
we have a counterpart of Proposition~\ref{prop:bsm1}. The answer to this question is affirmative:
\begin{proposition}
\label{prop:j1}
Let $f(S_T)$ satisfy \eqref{eq:claim} and~\eqref{eq:j_claim}.

(i) If $\pi=(x_0,\ffi)$ is an admissible strategy replicating $f(S_T)$, then $x_0\ge\EE f(S_T)$.

(ii) For any $x_0\ge\EE f(S_T)$, there is an admissible strategy $\pi=(x_0,\ffi)$ replicating $f(S_T)$.
\end{proposition}

\begin{proof}
The proof is similar to that of Proposition~\ref{prop:bsm1}.
The only difference is that now we need to construct, for any $x>0$,
a predictable process $\wt\psi$ with $\int_0^T|\wt\psi_u|\,\dd u<\infty$ $\PP$-a.s.
such that the local martingale
\begin{equation}
\label{eq:j_p1}
M_t:=x+\int_0^t\wt\psi_u\,\dd X_u,\quad t\in[0,T],
\end{equation}
is nonnegative and $M_T=0$ $\PP$-a.s.
For $n\in\bbN$ let us set
$$
\tau_n=T\wedge\inf\{t\in[0,T]\colon N_t=n\}\qquad(\inf\emptyset:=\infty).
$$
We define
\begin{equation}
\label{eq:j_p2}
\wt\psi_t=I_{\{0\le t\le\tau_1\}}\frac x{\lambda T}
+\sum_{n=1}^\infty I_{\{\tau_n<t\le\tau_{n+1}\}}\frac x{\lambda T}
\prod_{j=1}^n\left(1+\frac1{\lambda(T-\tau_j)}\right),\quad t\in[0,T].
\end{equation}
Then $\wt\psi$ is a finite piecewise constant predictable process on $[0,T]$,
hence $\int_0^T|\wt\psi_u|\,\dd u<\infty$ $\PP$-a.s.
A straightforward computation implies that the local martingale $M$ starting from $x$
defined in~\eqref{eq:j_p1} via $\wt\psi$ from~\eqref{eq:j_p2} is nonnegative and we have
$M_T=0$ $\PP$-a.s. This construction is illustrated in Figure~\ref{fig:j1}.
\end{proof}

\begin{figure}
\hspace{-0.6cm}
\input{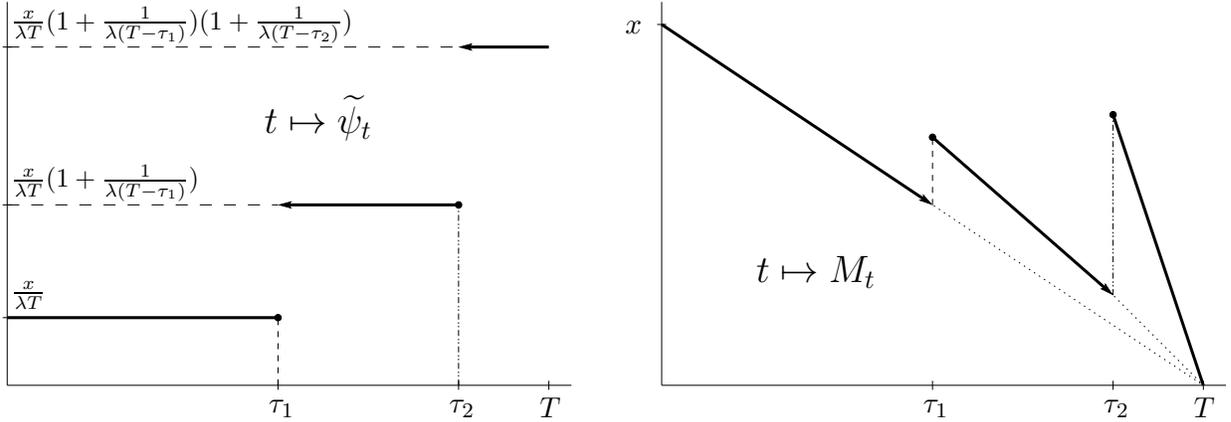}
\caption{\footnotesize{A schematic representation of typical paths
of the integrand $\wt\psi$ and the local martingale
$M=x+\int_0^\cdot\wt\psi_u\,\dd X_u$.
}}
\label{fig:j1}
\end{figure}

\begin{remark}
\label{rem:j1}
In general semimartingale models admissibility condition~\eqref{eq:adm}
required from strategies is needed to preclude arbitrage via strategies
that achieve positive terminal cashflows by allowing unbounded losses in the meantime.
One may wonder whether in this model based on the Poisson process,
arbitrage is excluded automatically without imposing~\eqref{eq:adm} on strategies.
The answer is negative: even in this simple model one should require~\eqref{eq:adm},
for otherwise the strategy $\pi:=(0,\ffi)$, where $\ffi:=-\frac{\wt\psi}{\sigma S_-}$
with $\wt\psi$ given by~\eqref{eq:j_p2}, starts from zero initial capital
and satisfies $V^\pi_T=x>0$ $\PP$-a.s. (in fact, $V^\pi=x-M$ with $M$ given by~\eqref{eq:j_p1}).
Such strategies are ruled out by the admissibility condition.
\end{remark}

Thus, any contingent claim $f(S_T)$ satisfying \eqref{eq:claim} and~\eqref{eq:j_claim}
can be replicated via different admissible strategies
with different initial capitals, but the discounted value process $U$ of the strategy $\pi=(U_0,\ffi)$
constructed in \eqref{eq:j_d1}--\eqref{eq:j_d3} starts from the minimal possible initial capital $U_0=\EE f(S_T)$.
Now we are going to express the  processes $U$ and $\ffi$ through the stock price process~$S$.

We have
$$
U_t=\EE(f(S_T)|\cF_t)=\EE\left(f\left(S_te^{\alpha(N_T-N_t)-\beta(T-t)}\right)|\cF_t\right)=V(t,S_t)\quad\PP\text{-a.s.},
$$
where
\begin{equation}
\label{eq:j_f1}
V(t,x)=\EE f\left(xe^{\alpha(N_T-N_t)-\beta(T-t)}\right)
=e^{-\lambda(T-t)}\sum_{k=0}^\infty f\left(xe^{\alpha k+\beta t-\beta T}\right)\frac{(\lambda(T-t))^k}{k!}.
\end{equation}
Let us note that the series in the right-hand side of~\eqref{eq:j_f1} converges absolutely for $x$ belonging
to the support of the distribution of $S_t$ because $f(S_T)\in L^1$
and hence $\EE\big(|f(S_T)|\big|\cF_t\big)<\infty$ $\PP$-a.s.
Even though, for a fixed $x$, the function $t\mapsto V(t,x)$ is not necessarily cadlag
(recall that $f$ is just a Borel function $(0,\infty)\to\bbR$
satisfying~\eqref{eq:claim} and~\eqref{eq:j_claim}),
$\PP$-a.s. it holds:
\begin{equation}
\label{eq:j_f2}
t\mapsto V(t,S_t)\quad\text{is cadlag and}\quad\lim_{u\nearrow t}V(u,S_u)=V(t,S_{t-}),\quad t\in(0,T].
\end{equation}
Indeed, substituting $x$ with $S_t$ in~\eqref{eq:j_f1} we get
\begin{equation}
\label{eq:j_f2.5}
V(t,S_t)=e^{-\lambda(T-t)}\sum_{k=0}^\infty f\left(S_0e^{\alpha k+\alpha N_t-\beta
T}\right)\frac{(\lambda(T-t))^k}{k!}.
\end{equation}
A similar formula for $V(t,S_{t-})$ holds with ``$N_{t-}$'' replacing ``$N_t$''.
Then \eqref{eq:j_f2} follows by noting that $\PP$-a.s. the process $N$ is cadlag and piecewise constant,
and that $\PP$-a.s. the series in the right-hand side of~\eqref{eq:j_f2.5} converges uniformly in $t\in[0,T]$.
The latter (technical) point is formally justified in Appendix~\ref{app:a}.

\begin{proposition}
\label{prop:j2}
A version of the second component $\ffi$
of the admissible strategy $\pi=(U_0,\ffi)$
replicating a contingent claim $f(S_T)$
from the minimal possible initial capital $U_0=\EE f(S_T)$
is given by the formula
\begin{equation}
\label{eq:j_f3}
\ffi_t=\frac{V(t,(1+\sigma)S_{t-})-V(t,S_{t-})}{\sigma S_{t-}},\quad t\in[0,T].
\end{equation}
Any other version is a predictable process that is equal to the right-hand side of~\eqref{eq:j_f3}
$\PP\times\mu_L$-a.e. on $\Omega\times[0,T]$, where $\mu_L$ denotes the Lebesgue measure.
\end{proposition}

\begin{remark}
\label{rem:j2}
As a consequence of this proposition we also get the following statement,
which is of interest when we work with specific contingent claims $f(S_T)$
and corresponding functions $V(t,x)$ (cf.~with the end of this section).
The process
\begin{equation}
\label{eq:j_f3.5}
\int_0^\cdot\left(V(u,(1+\sigma)S_{u-})-V(u,S_{u-})\right)\,\dd X_u,
\end{equation}
which is a priori only a local martingale, is in fact always a true martingale bounded from below
because, by Proposition~\ref{prop:j2}, it equals $U-U_0$ with $U$ given by~\eqref{eq:j_d1}.
Furthermore, it is a bounded martingale whenever the contingent claim $f(S_T)$ is bounded.
\end{remark}

\begin{proof}
Let $\ffi$ be a predictable process from~\eqref{eq:j_d3}.
Set $\wt\ffi=\sigma S_-\ffi$, i.e. $\wt\ffi$ is a predictable process
with $\int_0^T|\wt\ffi_u|\,\dd u<\infty$ $\PP$-a.s. and \eqref{eq:j_d2} holds.
Then $\PP$-a.s. we have
\begin{equation}
\label{eq:j_f4}
V(t,S_t)-V(t,S_{t-})=\Delta U_t=\wt\ffi_t\Delta X_t,\quad t\in[0,T].
\end{equation}
Since $\PP$-a.s. it holds: $\Delta X$ only takes values $0,1$
and, for $t\in[0,T]$, $S_t=e^\alpha S_{t-}=(1+\sigma)S_{t-}$ when $\Delta X_t=1$,
then the process
$$
\wt\psi_t:=V(t,(1+\sigma)S_{t-})-V(t,S_{t-}),\quad t\in[0,T],
$$
is a natural candidate for the process $\wt\ffi$.
Indeed, by~\eqref{eq:j_f2}, $\wt\psi$ is an adapted left-continuous process
with finite right-hand limits, hence $\wt\psi$ is predictable and locally bounded.
Consequently, $\int_0^T|\wt\psi_u|\,\dd u<\infty$ $\PP$-a.s.
and the integral $\int_0^\cdot\wt\psi_u\,\dd X_u$ is well-defined.
It remains to prove that $\PP\times\mu_L(\wt\ffi\ne\wt\psi)=0$.

Let us set $\wt\chi=\wt\ffi-\wt\psi$ and $Y=\int_0^\cdot\wt\chi_u\,\dd X_u$.
Since $\int_0^T|\wt\chi_u|\,\dd u<\infty$ $\PP$-a.s., the local martingale
$Y$ is the difference of the Lebesgue--Stieltjes integrals
$\int_0^\cdot\wt\chi_u\,\dd N_u-\int_0^\cdot\wt\chi_u\lambda\,\dd u$.
As a local martingale of finite variation $Y$ is a purely discontinuous local martingale
(see~\cite[Lem.~I.4.14~b)]{JacodShiryaev:03}).
But $Y$ does not have jumps because both $\wt\ffi$ and $\wt\psi$ satisfy~\eqref{eq:j_f4}.
Hence, by~\cite[Th.~I.4.18]{JacodShiryaev:03}, $Y_t=0$~$\PP$-a.s., $t\in[0,T]$. We have
\begin{equation}
\label{eq:j_f5}
\int_0^t\wt\chi_u\,\dd N_u=\int_0^t\wt\chi_u\lambda\,\dd u\quad\PP\text{-a.s.},\quad t\in[0,T].
\end{equation}
Since in the left-hand side of~\eqref{eq:j_f5} we have a pure jump process and in the right-hand side
a continuous process, both of them vanish $\PP$-a.s., and we get $\PP\times\mu_L(\wt\chi\ne0)=0$.
This concludes the proof.
\end{proof}

Proposition~\ref{prop:j1} implies that the admissible strategy
$\pi=(U_0,\ffi)$ with the discounted value process $U$, given in \eqref{eq:j_d1} and~\eqref{eq:j_d3},
replicates the contingent claim $f(S_T)$ from the minimal possible initial capital $U_0=\EE f(S_T)$.
Now we are going to compare the strategy $\pi$ with the delta-hedging strategy $\rho=(U_0,\delta)$,
$\delta_t=\frac{\partial V}{\partial x}(t,S_{t-})$,
and their respective discounted value processes $V^\pi$ and~$V^\delta$.
However, in order to make such a comparison, we need to assume additional regularity of the function $f$,
because in general there is no reason for the value function $V$ in~\eqref{eq:j_f1}
to be differentiable in the space variable.

Let a Borel measurable function 
$f\colon(0,\infty)\to\bbR$
satisfy~\eqref{eq:claim}
and~\eqref{eq:j_claim}
and assume that the following two assumptions hold:
\begin{eqnarray}
\label{eq:C_one}
f\in C^1(0,\infty)
\end{eqnarray}
and for every 
$x_0\in(0,\infty)$
there exists
$\eps\in(0,x_0)$
such that 
\begin{eqnarray}
\label{eq:Unif_Abs_Conv}
\sum_{k=0}^\infty \left|f'\left(xe^{\alpha k}\right)\right|\frac{(e^\alpha\lambda T)^k}{k!}\quad
\text{converges uniformly for}\quad x\in(x_0-\eps,x_0+\eps).
\end{eqnarray}
The parameter
$\alpha$
in~\eqref{eq:Unif_Abs_Conv}
is defined in~\eqref{eq:j_a1}
and 
$\lambda$
is the intensity of the Poisson process~$N$.
The following result
will allow us to define the delta-hedging strategy.

\begin{proposition}
\label{prop:j3}
Assume that a function
$f\colon(0,\infty)\to\bbR$
satisfies \eqref{eq:claim}, \eqref{eq:j_claim}, \eqref{eq:C_one}
and~\eqref{eq:Unif_Abs_Conv}.
Then the partial derivative
$\frac{\partial V}{\partial x}(t,x)$
of the function $V\colon[0,T]\times(0,\infty)\to\bbR$,
defined in~\eqref{eq:j_f1},
exists and is, for any
$(t,x)\in[0,T]\times(0,\infty)$,
given by the formula
\begin{eqnarray*}
\frac{\partial V}{\partial x}(t,x) & = & 
e^{-(\lambda+\beta) (T-t)} \sum_{k=0}^\infty f'\left(x e^{\alpha k
+\beta t-\beta T}\right)\frac{(e^\alpha \lambda(T-t))^k}{k!}.
\end{eqnarray*}
Furthermore the function
$\frac{\partial V}{\partial x}\colon[0,T]\times(0,\infty)\to\bbR$
is continuous.
\end{proposition}

In order to prove Proposition~\ref{prop:j3},
we recall the following known lemma from analysis.
Note that the sequence of functions $(F_n)_{n\in\bbN}$ below
is only assumed to converge at a single point.

\begin{lemma}
\label{lem:j1}
Let 
$F_n\in C^1(0,\infty)$
for 
$n\in\bbN$
and assume
$F_n(x_0)\to A\in\bbR$
for some
$x_0>0$.
If there exists 
$G\colon(0,\infty)\to\bbR$ 
such that the sequence 
$(F_n')_{n\in\bbN}$
converges locally uniformly to 
$G$,
then there exist 
$F\in C^1(0,\infty)$
such that 
$(F_n)_{n\in\bbN}$
converges 
locally uniformly to
$F$
and 
$F'=G$.
\end{lemma}

\begin{proof}[Proof of Proposition~\ref{prop:j3}]
Fix any
$t\in[0,T]$
and
define a sequence of functions
$$
F_n(x) = 
e^{-\lambda(T-t)}\sum_{k=0}^n f\left(xe^{\alpha k+\beta t-\beta
T}\right)\frac{(\lambda(T-t))^k}{k!},\quad x\in(0,\infty),\> n\in\bbN.
$$
Since 
$T-t\leq T$,
condition~\eqref{eq:Unif_Abs_Conv}
implies that the function
$G\colon(0,\infty)\to\bbR$,
given by
$G(x):=\lim_{n\to\infty} F'_n(x)$,
is well-defined
and that the sequence 
$(F'_n)_{n\in\bbN}$
converges locally uniformly to 
$G$
on the interval
$(0,\infty)$.
Assumptions~\eqref{eq:claim} and~\eqref{eq:j_claim}
imply that 
$f(S_T)\in L^1$
and hence we find
$$
e^{-\lambda T}\sum_{k=0}^\infty\left| f\left(S_0 e^{\alpha k-\beta T}\right)\right|
\frac{(\lambda T)^k}{k!} =\EE |f(S_T)| <\infty.
$$
This implies that for 
$x_0=e^{-\beta t}S_0\in(0,\infty)$,
the sequence 
$(F_n(x_0))_{n\in\bbN}$
converges to a real number. 
Lemma~\ref{lem:j1} now implies differentiability of $V$ in $x$
and the formula for $\frac{\partial V}{\partial x}(t,x)$.

To prove the continuity of the function
$\frac{\partial V}{\partial x}$
pick
any
$(s_0,y_0)\in[0,T]\times(0,\infty)$.
Let 
$\eps>0$
be as in~\eqref{eq:Unif_Abs_Conv}
for
$x_0:=y_0e^{\beta s_0-\beta T}$.
Then there exists 
$\eps'>0$
such that 
any
$(s,y)\in B_{\eps'}(s_0,y_0)$,
where the ball 
is given by
$B_{\eps'}(s_0,y_0)=\{(s',y')\in [0,T]\times(0,\infty)\colon
\max\{|s'-s_0|,|y'-y_0|\}<\eps'\}$,
satisfies
$|x_0-ye^{\beta s-\beta T}|<\eps$.
Assumption~\eqref{eq:Unif_Abs_Conv}
and the fact
$T-s\leq T$
for all 
$(s,y)\in B_{\eps'}(s_0,y_0)$
imply that the series
\begin{eqnarray*}
\sum_{k=0}^\infty \left|f'\left(ye^{\alpha k+\beta s-\beta T}\right)\right|\frac{(e^\alpha\lambda (T-s))^k}{k!}\quad
\text{converges uniformly in}\quad (s,y)\in B_{\eps'}(s_0,y_0).
\end{eqnarray*}
This implies joint continuity of $\frac{\partial V}{\partial x}$.
\end{proof}

Under the assumptions of Proposition~\ref{prop:j3},
we define the \emph{delta-hedging strategy} for the contingent claim $f(S_T)$
to be $\rho=(U_0,\delta)$ with $U_0=\EE S_T$ (cf.~with Proposition~\ref{prop:j1}) and
\begin{equation}
\label{eq:j_g1}
\delta_t=\frac{\partial V}{\partial x}(t,S_{t-}),\quad t\in[0,T].
\end{equation}
Let us note that $\delta$ is $S$-integrable as an adapted left-continuous process
with finite right-hand limits (hence predictable and locally bounded).
Below we again use the notation $\pi=(U_0,\ffi)$
for the admissible strategy replicating $f(S_T)$
given in Proposition~\ref{prop:j2}.
It is clear from \eqref{eq:j_f3} and~\eqref{eq:j_g1} that
in general
the strategies $\rho$ and $\pi$
have different discounted value processes and, in particular,
the delta-hedging strategy $\rho$ does not replicate $f(S_T)$.

Let us set
\begin{align}
\label{eq:j_g2}
\wt\ffi_t&=\sigma S_{t-}\ffi_t=V(t,(1+\sigma)S_{t-})-V(t,S_{t-}),
&t&\in[0,T],& & \\
\label{eq:j_g3}
\wt\delta_t&=\sigma S_{t-}\delta_t=\sigma S_{t-}\frac{\partial V}{\partial x}(t,S_{t-}),
&t&\in[0,T],& & 
\end{align}
and note that $V^\pi=U_0+\int_0^\cdot\wt\ffi_u\,\dd X_u$,
$V^\rho=U_0+\int_0^\cdot\wt\delta_u\,\dd X_u$.
The discounted replication error of the delta-hedging strategy $\rho$ is thus given by the formula
\begin{align*}
V^\rho_T-f(S_T)&=\int_0^T\left(\wt\delta_u-\wt\ffi_u\right)\,\dd X_u\\
&=-\int_0^T\left(V(u,(1+\sigma)S_{u-})-V(u,S_{u-})
-\sigma S_{u-}\frac{\partial V}{\partial x}(u,S_{u-})\right)\,\dd X_u.
\end{align*}
In particular, we have $\EE(V^\rho_T-f(S_T))=0$ whenever $V^\rho$ is a true martingale
(note that $V^\pi$ is always a true martingale by Remark~\ref{rem:j2}).
However, this fact cannot serve as a partial justification
for using the delta-hedging strategy in models with jumps, 
because if holds only under the risk-neutral measure.
Under a real-world probability measure the process $X$ is typically not a local martingale,
and hence the replication error of the delta-hedging strategy is typically 
biased (i.e. has a non-zero mean).

Finally, let us compare the strategies $\pi$ and $\rho$ in the following examples.

\begin{example}
\label{ex:j1}
Consider a contingent claim given by $f(S_T)=\log S_T$.
Note that the payoff $\log S_T$ is of interest in practice
as it arises naturally in the context of the derivatives on variance
(see e.g.~\cite{Gatheral:06}).
Clearly, the assumptions of Proposition~\ref{prop:j3} are satisfied
(note that $S_T\ge S_0e^{-\beta T}>0$ $\PP$-a.s.).
By~\eqref{eq:j_f1}, $V(t,x)=\log x+g(t)$ with some continuous function~$g$
(which can be computed explicitly).
By \eqref{eq:j_g2} and~\eqref{eq:j_g3}, we have
$$
\wt\ffi_t\equiv\log(1+\sigma),\quad\wt\delta_t\equiv\sigma,\quad t\in[0,T],
$$
hence
$$
V^\pi_t=U_0+\left(\log(1+\sigma)\right)X_t,\quad V^\rho_t=U_0+\sigma X_t,\quad t\in[0,T].
$$
In particular, $\pi$ and $\rho$ are admissible, $V^\pi$ and $V^\rho$ are true martingales,
$\pi$ replicates $\log S_T$, but $\rho$, which is the delta-hedging strategy for $\log S_T$,
replicates an affine transformation of $\log S_T$.
\end{example}

\begin{example}
\label{ex:j2}
Consider $f(S_T)=S_T^a$ with some $a\in\bbR\setminus\{0\}$.
Again, the assumptions of Proposition~\ref{prop:j3} are satisfied.
By~\eqref{eq:j_f1}, $V(t,x)=x^ag(t)$ with some continuous function~$g$. We get
$$
\wt\ffi_t=\left((1+\sigma)^a-1\right)S_{t-}^ag(t),\quad\wt\delta_t=a\sigma S_{t-}^ag(t),\quad t\in[0,T].
$$
It follows that
\begin{equation}
\label{eq:j_g4}
V^\rho-U_0=\frac{a\sigma}{(1+\sigma)^a-1}(V^\pi-U_0).
\end{equation}
By Remark~\ref{rem:j2}, $V^\pi$ is a martingale bounded from below.
Then, by~\eqref{eq:j_g4}, $V^\rho$ is a martingale bounded from below
(note that $\frac{a\sigma}{(1+\sigma)^a-1}>0$ regardless of whether $a>0$ or $a<0$).
In particular, the delta-hedging strategy $\rho$ is admissible.
But it does not replicate $S_T^a$ if $a\ne1$.

Finally, let us note that in the case $a<0$ the payoff $S_T^a$ is bounded.
Hence by Remark~\ref{rem:j2} and the formula in~\eqref{eq:j_g4} we have an example, where both discounted value processes $V^\pi$ and $V^\rho$
are bounded martingales, but the delta-hedging strategy does not replicate the 
contingent claim.
\end{example}

%

\appendix
\section{}
\label{app:a}
Here we prove that under the condition $f(S_T)\in L^1$ we have that, for $\PP$-a.a.~$\omega$,
the series in the right-hand side of~\eqref{eq:j_f2.5} converges uniformly in $t\in[0,T]$.

Let $M\in(0,\infty)$ and $a_k\ge0$, $k=0,1,\ldots$, be deterministic numbers such that
\begin{equation}
\label{eq:a1}
\sum_{k=0}^\infty a_k\frac{M^k}{k!}<\infty.
\end{equation}
Let $(N_x)_{x\in[0,M]}$ be a deterministic nondecreasing cadlag function with $N_0=0$
that takes only integer values and such that $\Delta N_x$, $x\in[0,M]$,
can take only values $0$ and~$1$. In other words, $(N_x)$ is a typical path of a Poisson process.
Further let us define deterministic numbers
$$
m=N_M\quad\text{and}\quad\tau_1=M\wedge\inf\{x\in[0,M]\colon N_x=1\}\quad(\inf\emptyset:=\infty)
$$
and note that $\tau_1>0$.

\begin{lemma}
\label{lem:a1}
The series
\begin{equation}
\label{eq:a2}
\sum_{k=0}^\infty a_{k+N_{M-x}}\frac{x^k}{k!}
\end{equation}
converges uniformly in $x\in[0,M]$.
\end{lemma}

\begin{proof}
Let $R$ denote the radius of convergence of the power series (in~$x$)
\begin{equation}
\label{eq:a3}
\sum_{k=0}^\infty a_k\frac{x^k}{k!}.
\end{equation}
By~\eqref{eq:a1}, $R\ge M$.
The $m$ power series corresponding to $m$ derivatives of~\eqref{eq:a3} are
\begin{equation}
\label{eq:a4}
\sum_{k=0}^\infty a_{k+l}\frac{x^k}{k!},\quad l=1,2,\ldots,m,
\end{equation}
and they have the same radius of convergence $R\ge M$.
Then all $m+1$ series in \eqref{eq:a3} and~\eqref{eq:a4}
converge uniformly in $x\in[0,M-\eps]$ with $\eps:=\frac{\tau_1}2$,
hence the series in~\eqref{eq:a2} converges uniformly in $x\in[0,M-\eps]$.
It remains to note that~\eqref{eq:a2} converges uniformly also in $x\in[M-\eps,M]$
due to~\eqref{eq:a1} and the fact that $N_{M-x}=0$ for $x\in[M-\eps,M]$.
\end{proof}

Now using the notation of Section~\ref{sec:j}, we set $M=\lambda T$
and $a_k=|f(S_0e^{\alpha k-\beta T})|$, $k=0,1,\ldots$.
The condition~\eqref{eq:a1} amounts to $f(S_T)\in L^1$, which is satisfied.
Applying Lemma~\ref{lem:a1} pathwise we get the desired result.

\bibliographystyle{abbrv}
\bibliography{refs}

\begin{thebibliography}{10}

\bibitem{BinghamKiesel:04}
N.~H. Bingham and R.~Kiesel.
\newblock {\em Risk-neutral valuation}.
\newblock Springer Finance. Springer-Verlag London Ltd., London, second
  edition, 2004.
\newblock Pricing and hedging of financial derivatives.

\bibitem{Bjork:09}
T.~Björk.
\newblock {\em Arbitrage Theory in Continuous Time}.
\newblock Oxford University Press, third edition, 2009.

\bibitem{ContTankov:03}
R.~Cont and P.~Tankov.
\newblock {\em Financial Modelling With Jump Processes}.
\newblock Chapman \& Hall, 2003.

\bibitem{DanaJeanblanc:03}
R.-A. Dana and M.~Jeanblanc.
\newblock {\em Financial markets in continuous time}.
\newblock Springer Finance. Springer-Verlag, Berlin, 2003.
\newblock Translated from the 1998 French original by Anna Kennedy.

\bibitem{DavisPanasZariphopoulou:93}
M.~H.~A. Davis, V.~G. Panas, and T.~Zariphopoulou.
\newblock European option pricing with transaction costs.
\newblock {\em SIAM J. Control and Optimization}, 31(2):470--493, 1993.

\bibitem{DelbaenSchachermayer:94}
F.~Delbaen and W.~Schachermayer.
\newblock A general version of the fundamental theorem of asset pricing.
\newblock {\em Math. Ann.}, 300(3):463--520, 1994.

\bibitem{DelbaenSchachermayer:98}
F.~Delbaen and W.~Schachermayer.
\newblock The fundamental theorem of asset pricing for unbounded stochastic
  processes.
\newblock {\em Math. Ann.}, 312(2):215--250, 1998.

\bibitem{DelbaenSchachermayer:06}
F.~Delbaen and W.~Schachermayer.
\newblock {\em The mathematics of arbitrage}.
\newblock Springer Finance. Springer-Verlag, Berlin, 2006.

\bibitem{FollmerSchied:06}
H.~F{\"o}llmer and A.~Schied.
\newblock {\em Stochastic finance}, volume~27 of {\em de Gruyter Studies in
  Mathematics}.
\newblock Walter de Gruyter \& Co., Berlin, extended edition, 2004.
\newblock An introduction in discrete time.

\bibitem{Gatheral:06}
J.~Gatheral.
\newblock {\em The volatility surface: a practitoner's guide}.
\newblock John Wiley \& Sons, Inc., 2006.

\bibitem{HarrisonPliska:81}
J.~M. Harrison and S.~R. Pliska.
\newblock Martingales and stochastic integrals in the theory of continuous
  trading.
\newblock {\em Stochastic Process. Appl.}, 11(3):215--260, 1981.

\bibitem{HubalekKallsenKrawczyk:06}
F.~Hubalek, J.~Kallsen, and L.~Krawczyk.
\newblock Variance-optimal hedging for processes with stationary independent
  increments.
\newblock {\em Ann. Appl. Probab.}, 16(2):853--885, 2006.

\bibitem{HuntKennedy:04}
P.~J. Hunt and J.~E. Kennedy.
\newblock {\em Financial derivatives in theory and practice}.
\newblock Wiley Series in Probability and Statistics. John Wiley \& Sons Ltd.,
  Chichester, revised edition, 2004.

\bibitem{JacodShiryaev:03}
J.~Jacod and A.~N. Shiryaev.
\newblock {\em Limit {T}heorems for {S}tochastic {P}rocesses}, volume 288 of
  {\em Grundlehren der Mathematischen Wissenschaften}.
\newblock Springer-Verlag, Berlin, second edition, 2003.

\bibitem{JeanblancYorChesney:09}
M.~Jeanblanc, M.~Yor, and M.~Chesney.
\newblock {\em Mathematical methods for financial markets}.
\newblock Springer Finance. Springer-Verlag London Ltd., London, 2009.

\bibitem{KaratzasShreve:98}
I.~Karatzas and S.~E. Shreve.
\newblock {\em Methods of {M}athematical {F}inance}, volume~39 of {\em
  Applications of Mathematics (New York)}.
\newblock Springer-Verlag, New York, 1998.

\bibitem{LaiXing:08}
T.~L. Lai and H.~Xing.
\newblock {\em Statistical models and methods for financial markets}.
\newblock Springer Texts in Statistics. Springer, New York, 2008.

\bibitem{LambertonLapeyre:08}
D.~Lamberton and B.~Lapeyre.
\newblock {\em Introduction to stochastic calculus applied to finance}.
\newblock Chapman \& Hall/CRC Financial Mathematics Series. Chapman \&
  Hall/CRC, Boca Raton, FL, second edition, 2008.

\bibitem{Lipton:01}
A.~Lipton.
\newblock {\em Mathematical Methods for Foreign Exchange}.
\newblock World Scientific, 2001.

\bibitem{MusielaRutkowski:08}
M.~Musiela and M.~Rutkowski.
\newblock {\em Martingale methods in financial modelling}, volume~36 of {\em
  Stochastic Modelling and Applied Probability}.
\newblock Springer-Verlag, Berlin, second edition, 2005.

\bibitem{Schoutens:03}
W.~Schoutens.
\newblock {\em Levy Processes in Finance: Pricing Financial Derivatives}.
\newblock Wiley, 2003.

\bibitem{Schweizer:01}
M.~Schweizer.
\newblock A guided tour through quadratic hedging approaches.
\newblock In {\em Option pricing, interest rates and risk management}, Handb.
  Math. Finance, pages 538--574. Cambridge Univ. Press, Cambridge, 2001.

\bibitem{Shiryaev:99}
A.~N. Shiryaev.
\newblock {\em Essentials of Stochastic Finance}, volume~3 of {\em Advanced
  Series on Statistical Science \& Applied Probability}.
\newblock World Scientific Publishing Co. Inc., River Edge, NJ, 1999.
\newblock Facts, models, theory, Translated from the Russian manuscript by N.
  Kruzhilin.

\bibitem{Shreve:04}
S.~E. Shreve.
\newblock {\em Stochastic calculus for finance. {II}}.
\newblock Springer Finance. Springer-Verlag, New York, 2004.
\newblock Continuous-time models.

\bibitem{Tankov:10}
P.~Tankov.
\newblock {\em Pricing and hedging in exponential {L}\'evy models: review of
  recent results}.
\newblock Paris-Princeton Lecture Notes in Mathematical Financee.
  Springer-Verlag, Berlin, 2010.

\end{thebibliography}
\end{document}